\documentclass[conference]{IEEEtran}

\usepackage{amsmath, amsthm, amssymb}
\usepackage{latexsym}
\usepackage{graphicx}
\usepackage{verbatim}
\usepackage{mathrsfs}
\usepackage{float}
\usepackage[lofdepth, lotdepth]{subfig}
\usepackage{array}
\usepackage{url}
\usepackage{algorithm}
\usepackage[noend]{algorithmic}
\usepackage{cite}
\usepackage[dvipsnames]{xcolor}
\usepackage{enumerate}
\usepackage{eucal}
\usepackage{stmaryrd}
\usepackage{mathtools}
\usepackage{pgf}
\usepackage{tikz}
\usetikzlibrary{arrows,automata}
\usepackage[latin1]{inputenc}
\usetikzlibrary{automata,positioning}

\usepackage{mathtools}
\DeclarePairedDelimiter\ceil{\lceil}{\rceil}

\theoremstyle{plain}
\newtheorem{theorem}{Theorem}
\newtheorem{proposition}{Proposition}
\newtheorem{lemma}{Lemma}
\newtheorem{conjecture}{Conjecture}
\newtheorem{corollary}{Corollary}
\newtheorem{construction}{Construction}
\theoremstyle{definition}
\newtheorem{definition}{Definition}
\newtheorem{example}{Example}
\newtheorem{remark}{Remark}

\newcommand{\C}{{\mathcal C}}
\newcommand{\D}{{\mathcal D}}

\DeclareMathAlphabet{\mathbfsl}{OT1}{ppl}{b}{it} 

\newcommand{\bL}{\mathbfsl{L}}

\newcommand{\bX}{\mathbfsl{X}}

\newcommand{\bU}{\mathbfsl{U}}

\newcommand{\bp}{{\mathbfsl p}}
\newcommand{\bq}{{\mathbfsl q}}

\newcommand{\bu}{{\mathbfsl u}}
\newcommand{\bv}{{\mathbfsl v}}

\newcommand{\bs}{{\mathbfsl s}}
\newcommand{\by}{{\mathbfsl y}}

\newcommand{\bc}{{\mathbfsl c}}

\newcommand{\bx}{{\mathbfsl{x}}}
\newcommand{\bz}{{\mathbfsl{z}}}

\newcommand{\bsg}{{\boldsymbol{\sigma}}}

\renewcommand{\ge}{\geqslant}
\renewcommand{\le}{\leqslant}

\newcommand{\et}{{\emph{et al.}}}

\newcommand{\enc}{\textsc{Enc}}
\newcommand{\dec}{\textsc{Dec}}

\IEEEoverridecommandlockouts
\begin{document}

\pagestyle{empty}

\title{On the Design of Codes for DNA Computing: Secondary Structure Avoidance Codes\\[-3mm]}


\author{\IEEEauthorblockN{Tuan Thanh Nguyen\IEEEauthorrefmark{1},
Kui Cai\IEEEauthorrefmark{1},
Han Mao Kiah\IEEEauthorrefmark{2},
Duc Tu Dao\IEEEauthorrefmark{2},  
and Kees A. Schouhamer Immink\IEEEauthorrefmark{3}}\\[-3mm]
\IEEEauthorblockA{
\IEEEauthorrefmark{1}
Science, Mathematics and Technology Cluster, Singapore University of Technology and Design, Singapore 487372\\
\IEEEauthorrefmark{2}%
School of Physical and Mathematical Sciences, Nanyang Technological University, Singapore 637371\\
\IEEEauthorrefmark{3}%
Turing Machines Inc, Willemskade 15d, 3016 DK Rotterdam, The Netherlands\\
Emails: \{tuanthanh\_nguyen, cai\_kui\}@sutd.edu.sg, \{hmkiah,daoductu001\}@ntu.edu.sg, immink@turing-machines.com\\[-4mm]
}
}
\maketitle
\maketitle

\hspace{-3mm}\begin{abstract}

In this work, we investigate a challenging problem, which has been considered to be an important criterion in designing codewords for DNA computing purposes, namely {\em secondary structure avoidance} in single-stranded DNA molecules. In short, secondary structure refers to the tendency of a single-stranded DNA sequence to fold back upon itself, thus becoming inactive in the computation process. 
While some design criteria that reduces the possibility of secondary structure formation has been proposed by Milenkovic and Kashyap (2006), the main contribution of this work is to provide an explicit construction of DNA codes that completely avoid secondary structure of arbitrary stem length.

Formally, given codeword length $n$ and arbitrary integer $m\ge 2$, we provide efficient methods to construct DNA codes of length $n$ that avoid secondary structure of any stem length more than or equal to $m$. Particularly, when $m=3$, our constructions yield a family of DNA codes of rate 1.3031 bits/nt, while the highest rate found in the prior art was 1.1609 bits/nt. In addition, for $m\ge 3\log n + 4$, we provide an efficient encoder that incurs only one redundant symbol.



\end{abstract}

\section{Introduction}
DNA computing is an emerging branch of computing that uses DNA, biochemistry, and molecular biology hardware. 
The field of DNA computation started with the following demonstration by Adleman in 1994~\cite{adle:1994}. In this seminal experiment, Adleman solved an instance of the directed traveling salesperson problem by first representing each city with a synthetic DNA molecule. Then by allowing the strands to hybridize in a highly parallel fashion, Adleman obtained the desired solution. Since then, similar methods have been expanded to several attractive applications, including the development of storage technologies \cite{church2012,fountain2017,Organick:2018,goldman2013}, and cell-based computation systems for cancer diagnostics and treatment \cite{son:2004}. Recently, the hybridization process was exploited to allow random access in DNA data storage~\cite{yazdi:2017}.

In DNA computing, only short single-stranded DNA sequences (or {\em oligonucleotide sequences}) are used, where each of them is an oriented word consisting of four bases (or {\em nucleotides}): Adenine ({\bf A}), Thymine  ({\bf T}), Cytosine ({\bf C}), and Guanine ({\bf G}). A set of encoded DNA sequences (also called DNA codewords), that satisfies certain special properties (or {\em constraints}) for DNA computing purposes, is called a DNA code. A broad description of the kinds of constraint problems that arise in coding for DNA computing was introduced by Milenkovic and Kashyap in 2006 \cite{O2:2005}, including {\em constant {\bf G}{\bf C}-content constraint} (refers to the percentage of nucleotides that are either {\bf G} or {\bf C}), {\em Hamming distance constraint} (that requires DNA codewords to be sufficiently different among themselves), and {\em secondary structure formation avoidance constraint} (that prevents DNA sequence to fold back upon itself, and consequently becoming inactive in the computation process). 
Similar considerations were described in~\cite{yazdi:2018, chee:2020} for the design of primer address sequences in random access of DNA-based data storage systems.
While constant {\bf G}{\bf C}-content constraint and Hamming distance constraint have been extensively investigated \cite{K:2021, KDG:2021, O:2005, O2:2005, TT:2021, KING:2003, KING:2005, segmented}, the study for secondary structure avoidance is much less profound.


For a DNA sequence, a secondary structure is formed by a chemically active to fold back onto itself by complementary base pair hybridization (illustrated via Figure 1). Here, the {\em Watson-Crick complement} is defined as: $\overline{{\bf A}}={\bf T}, \overline{\bf T}={\bf A}, \overline{\bf C}={\bf G}$, and $\overline{\bf G}={\bf C}$. For a sequence $\bx=x_1x_2x_3\ldots x_{n-1}x_n$ over the DNA alphabet $\D=\{{\bf A}, {\bf T}, {\bf C}, {\bf G}\}$, the {\em reverse-complement} of $\bx$ is defined as ${\rm RC}({\bx})={\overline{x_n}} \text{ } \overline{x_{n-1}}\ldots  \overline{x_3}  \text{ } \overline{x_2}  \text{ } \overline{x_1}$. In Figure 1, sub-sequences ${\color{ForestGreen}{\bx={\bf A}{\bf T}{\bf A}{\bf C}{\bf C}}}$ and ${\color{blue}{\by={\rm RC}({\bx})={\bf G}{\bf G}{\bf T}{\bf A}{\bf T}}}$ of the DNA sequence $\sigma$ bind to each other after pairing of {\bf A} with {\bf T} and {\bf G} with {\bf C}, forming a secondary structure with a loop and a {\em stem} of length 5. DNA sequences with secondary structures are less active in the computation process \cite{O2:2005}, and hence, before reading such sequences in a wet lab, they need to be unfolded, costing more resources and energy. There exist some simple dynamic programming techniques \cite{Bres:1986,RN:1980} that can approximately predict the secondary structures in a given DNA sequence (for example, see the Nussinov-Jacobson (NJ) algorithm in \cite{RN:1980} as one of the most widely used schemes). Based on the NJ algorithm, the authors in \cite{O2:2005, O:2005} found some design criteria that reduce the possibility of secondary structure formation in a codeword. A natural question is whether there exists efficient design of DNA constrained codes that avoid the formation of secondary structures. 


\begin{figure*}[h]
\begin{center}
\includegraphics[width=9cm]{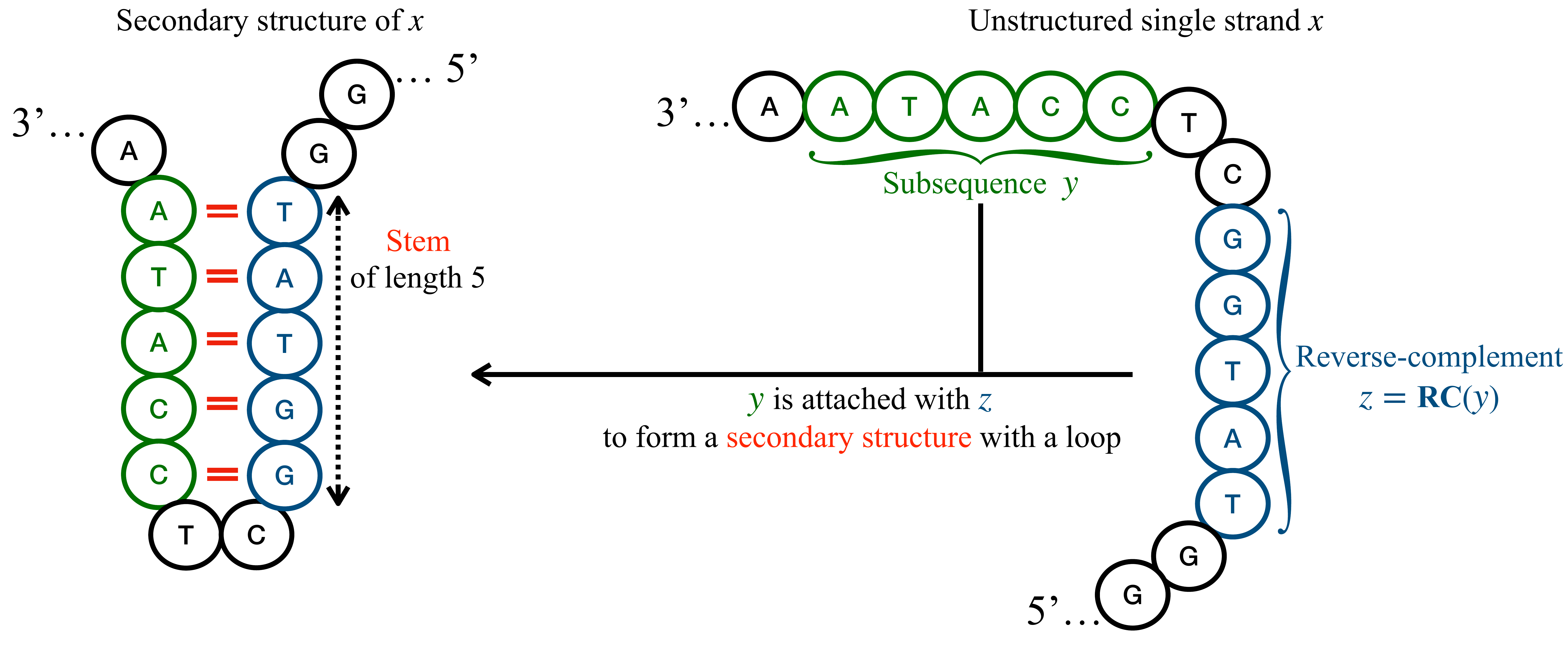}
\end{center}
\caption{DNA secondary structure model. Here, the Watson-Crick complement is: $\overline{{\bf A}}={\bf T}, \overline{\bf T}={\bf A}, \overline{\bf C}={\bf G}$, and $\overline{\bf G}={\bf C}$.}
\label{fig1}
\end{figure*}


It has been shown experimentally that the number of base pairs in stem regions (or {\em stem length}) is one important factor influencing the secondary structure of a DNA sequence. Given codeword length $n$ and an integer $m\ge 2$, we study the problem of constructing DNA codes of length $n$ that avoid secondary structure of any stem length more than or equal to $m$. To the best of our knowledge, this work is the first attempt aimed at providing a rigorous solution for DNA codes avoiding secondary structure for general stem lengths.

\section{Preliminary}\label{sec:prelim}
In this work, we use $\D$ to denote the DNA alphabet, where $\D=\{{\bf A}, {\bf T}, {\bf C}, {\bf G}\}$. Here, we have the Watson-Crick complement where $\overline{{\bf A}}={\bf T}, \overline{\bf T}={\bf A}, \overline{\bf C}={\bf G}$, and $\overline{\bf G}={\bf C}$. 

Given two sequences $\bx$ and $\by$, we let $\bx\by$ denote the {\em concatenation} of the two sequences. 

Throughout this work, given a sequence $\bx$ of length $n$, we say $\by$ is a subsequence of length $k$ of $\bx$, where $k\le n$, if $\by=x_ix_{i+1}\ldots x_{i+k-1}$ for some $1\le i\le n-k+1$. In other words, we only consider the subsequences including consecutive symbols in $\bx$. Two subsequences $\by$ and $\bz$ of $\bx$ are said to be {\em non-overlapping} if we have $\by=x_ix_{i+1}\ldots x_{i+k-1}$, $\bz=x_jx_{j+1}\ldots x_{j+\ell-1}$, where $i>j+\ell-1$ or $j>i+k-1$.


\begin{definition}
For a DNA sequence $\bx \in \D^n$, $\bx=x_1x_2\ldots x_n$, the reverse-complement of $\bx$, is defined as ${\rm RC}({\bx})={\overline{x_n}} \text{ } \overline{x_{n-1}}\ldots  \overline{x_3}  \text{ } \overline{x_2}  \text{ } \overline{x_1}$. 
\end{definition}

\begin{definition}
Given $0<m\le n$, a DNA sequence $\bx\in \D^n$ is said to be $m$-secondary structure avoidance (or $m$-SSA) sequence if for all $k\ge m$, there does not exist any pair of non-overlapping subsequences $\by, \bz$ of length $k$ of $\bx$ such that $\by={\rm RC}({\bz})$. A code $\C$ is said to be an $(n,\D;m)$ SSA code if for every codeword $\bx\in \C \cap \D^n$, we have $\bx$ is $m$-SSA. 
\end{definition}

The following result is immediate. 
\begin{lemma}\label{lemma1}
Given $m,n>0$, if a sequence $\bx\in \D^n$ is  $m$-SSA then $\bx$ is $m'$-SSA for all $m'>m$. 
\end{lemma}


For a code $\C \subseteq \D^n$, the {\em code rate} is measured by the value $\log |\C|/n$. Intuitively, it measures the number of information bits stored in each DNA symbol.  Suppose that we have an infinite family of codes $\{\C_n\}_{n=1}^\infty$, where $\C_n$ is a code of length $n$, then the asymptotic rate of the family is ${\bf r} \triangleq \lim_{n\to \infty} \frac{\log |\C_n|}{n}$. Here, we adopt the notation $\log$ to mean logarithm base two. %


\begin{definition} 
Given $m>0$, for $n>0$, let ${\rm A}(n,\D;m)$ be the total number of DNA sequences of length $n$ that are $m$-SSA. The channel capacity, denoted by ${\rm c}_m$, is defined by: 
\begin{small}
\begin{equation*} 
{\rm c}_m=\lim_{n \to \infty} \frac{\log {\rm A}(n,\D;m)}{n}.
\end{equation*}  
\end{small}
\end{definition}

The following result is immediate. 

\begin{lemma}\label{trivial-bound} 
Given $m>0$, let $S_m$ be the set of all DNA sequences of length $m$ such that, there is no pair of sequences $\by, \bz \in S_m$, not necessary distinct, such that $\by={\rm RC}(\bz)$. We then have ${\rm c}_m \le 1/m \log |S_m|$. 
\end{lemma}

Observe that the size of $S_m$ can be computed easily for constant $m$, a trivial upper bound is that $|S_m| \le 4^m/2$, and consequently, we obtain ${\rm c}_2 \le 1.5$ and ${\rm c}_3 \le 1.67$. 

To construct an $(n,\D;m)$ SSA code for arbitrary $m>0$ by concatenation method, one can find the largest set $S_N$ for some suitable value of $N$, such that, for $n=Nk$, each codeword is a concatenation of $k$ sequences of length $N$ from $S_N$ and each concatenation does not create a reverse-complement subsequence from previous concatenations. The construction yields a family of DNA codes of rate $1/N \log |S_N|$ bits/nt. For example, for $m=3$, Krishna Gopal Benerjee and Adrish Banerjee \cite{K:2021} constructed an $(n,\D;3)$ SSA code via such a set $S=\{{\bf A}{\bf A}, {\bf C}{\bf C}, {\bf A}{\bf C}, {\bf C}{\bf A},{\bf T}{\bf C}\}$. 


\begin{theorem}[Benerjee and Banerjee \cite{K:2021}]
Set $S=\{{\bf A}{\bf A}, {\bf C}{\bf C}, {\bf A}{\bf C}, {\bf C}{\bf A},{\bf T}{\bf C}\}$. Let $\C$ be the DNA code of length $2n$ where each codeword is a concatenation of words of length two from $S$. We then have $\C$ is an $(n,\D;3)$ SSA code, i.e. every codeword of $\C$ is $3$-SSA. The size of the code is $|\C|=5^n$, and the code rate is $1/2 \log 5 = 1.1609$ bits/nt.  
\end{theorem}

\subsection{Paper Organisation and Our Main Contribution}

Since the number of base pairs in stem regions (or stem length) is one important factor influencing the secondary structure of a DNA sequence, this work aims at providing a rigorous solution for $(n,\D;m)$ SSA codes given arbitrary $m$. 
The paper is organised as follows. 
\begin{itemize}
\item Section III presents two efficient constructions of  $(n,\D;m)$ SSA codes for arbitrary $m>0$. The first construction is based on {\em block concatenation}, which concatenates blocks of fixed length $m$ from a predetermined set. On the other hand, crucial to the second construction is the concept of {\em symbol-composition constrained codes}. Particularly, when $m = 3$, the second construction yields a family of DNA codes of rate $1.3031$ bits/nt, which is higher than the code rate in \cite{K:2021}. 
\item Section IV presents a linear-time encoding method for $(n,\D;m)$ SSA code with only one redundant symbol whenever $m\ge 3 \log n+4$. The coding method is based on {\em sequence replacement technique}. 

\end{itemize}

\section{Constructions of $(n,\D;m)$ SSA Codes for arbitrary $m>0$} 
The first method is based on block concatenation, which concatenates blocks of length $m$ from a predetermined set.
\subsection{Constructions via Block Concatenation}
\begin{construction}\label{fixed-concatenate}
Given $m>0$, $n=mk$ for some integer $k>0$, set $t=\ceil{m/3}$. Let $S_m^*$ be the set of all DNA sequences of length $m$ such that for any pair of sequences $\bx_1,\bx_2 \in S_m^*$, not necessary distinct, there is no pair of subsequences $\by$ of $\bx_1$ and $\bz$ of $\bx_2$ of length $t$ such that $\by={\rm RC}(\bz)$. Let $\C$ be the DNA code of length $n$, where each codeword is a concatenation of $k$ sequences of length $m$ in $S_m^*$.
\end{construction} 

\begin{theorem}\label{con1}
The constructed code $\C$ from Construction~\ref{fixed-concatenate} is an $(n,\D;m)$ SSA code. 
\end{theorem} 

\begin{proof}
We prove the correctness of Theorem~\ref{con1} by contradiction. Suppose that, there exists a codeword $\bc\in \C, \bc=\bx_1\bx_2\ldots \bx_k$, where $\bx_i \in S_m^*$, and $\bc$ is not $m$-SSA. In other words, there exists two non-overlapping subsequences $\by$,$\bz$ of $\bc$ of length $m'\ge m$ such that $\by={\rm RC}(\bz)$. 

Suppose that $\by=Y_1 Y_2$ where $Y_1$ is a subsequence of $\bx_i$, and $Y_2$ is a subsequence of $\bx_{i+1}\bx_{i+2}\ldots \bx_{i+h}$ for some $h\ge 1$. We have $\bz={\rm RC}(Y_2) {\rm RC}(Y_1)$. The trivial case is if $h>1$, or $Y_2$ is of length more than $m$, then $\bx_{i+1}$ is a subsequence of $Y_2$ and $ {\rm RC}(\bx_{i+1})$ is a subsequence of $\bz$. Clearly, if $ {\rm RC}(\bx_{i+1}) \equiv \bx_j$, we have a contradiction. On the other hand, if $ {\rm RC}(\bx_{i+1})=W_1 W_2$ where $W_1$ is a subsequence of $\bx_j$ and $W_2$ is a subsequence of $\bx_{j+1}$ for some $j$, then at least one subsequence $W_1$ or $W_2$ is of size more than $t$, we also have a contradiction. We conclude that $h=1$, or $Y_2$ is simply a subsequence of $\bx_{i+1}$.

Now, since $\by=Y_1Y_2$ is of length $m'\ge m$, at least $Y_1\ge t$ or $Y_2\ge t$. W.l.o.g, assume that  $Y_1\ge t$.

We observe that ${\rm RC}(Y_1)$ cannot be a subsequence of any $\bx_j$ by Construction 1. In other words, ${\rm RC}(Y_1)=W_1 W_2$ where $W_1$ is a subsequence of $\bx_j$ and $W_2$ is a subsequence of $\bx_{j+1}$ for some $j$. Similarly, we observe that the length of $W_1, W_2$ must be strictly smaller than $t$, otherwise, for example, if the length of $W_1$ is more than or equal to $t$, then two sequences $\bx_i$ and $\bx_j$ in $S_m^*$ contain ${\rm RC}(W_1)$ and $W_1$ as subsequences, we have a contradiction. Since both the length of $W_1, W_2$ must be strictly smaller than $t$, causing the length of $Y_1$ is smaller than $2t$, we conclude that the length of $Y_2$ is at least $t$. 

Now, let $\bU={\rm RC}(Y_2) \cap \bx_{j+1}$, the subsequence that belongs to both $\bx_{j+1}$ and ${\rm RC}(Y_2)$, which is of size at least $t$. We then have $\bU$ is a subsequence of $\bx_{j+1}$ while ${\rm RC}(\bU)$ is a subsequence of ${\rm RC}({\rm RC}(Y_2))=Y_2$, a subsequence of $\bx_{i+1}$. We then have a contradiction. 

In conclusion, we have $\C$ is an $(n,\D;m)$ SSA code. We highlight our proof sketch of Theorem~\ref{con1} in Figure~\ref{fig2}.
\end{proof}
\begin{figure}[h]
\caption*{Claim 1: When $|Y_1|\ge t$, we observe that ${\rm RC}(Y_1)$ cannot be a subsequence of any $\bx_{j+1}$.}
\begin{center}
\includegraphics[width=5.5cm]{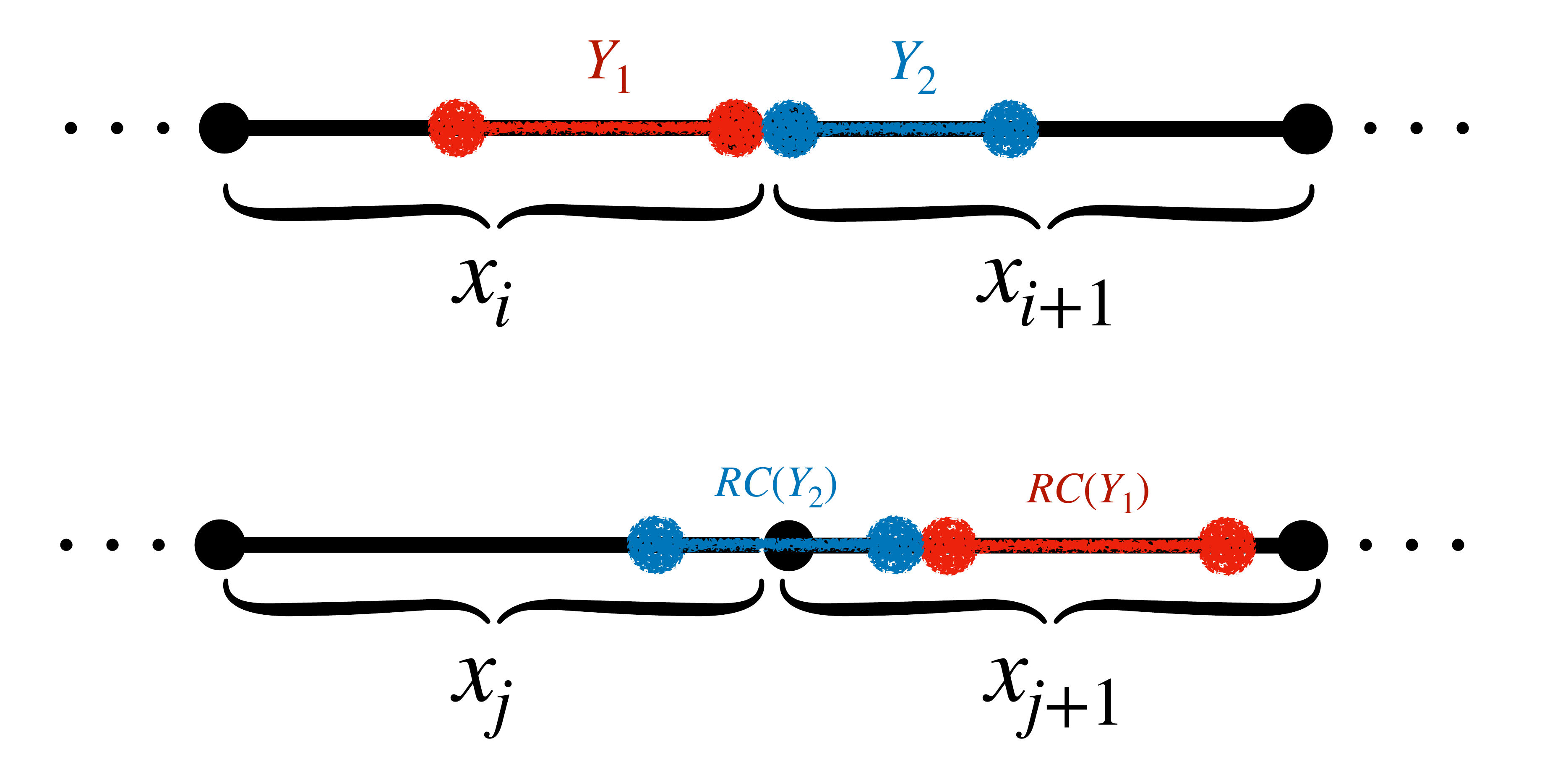}
\end{center}
\caption*{Claim 2: When $Y_1=W_1W_2$ and ${\rm RC}(Y_1)={\rm RC}(W_2){\rm RC}(W_1)$, we must have $|W_1|\le t, |W_2|\le t$.}
\begin{center}
\includegraphics[width=6cm]{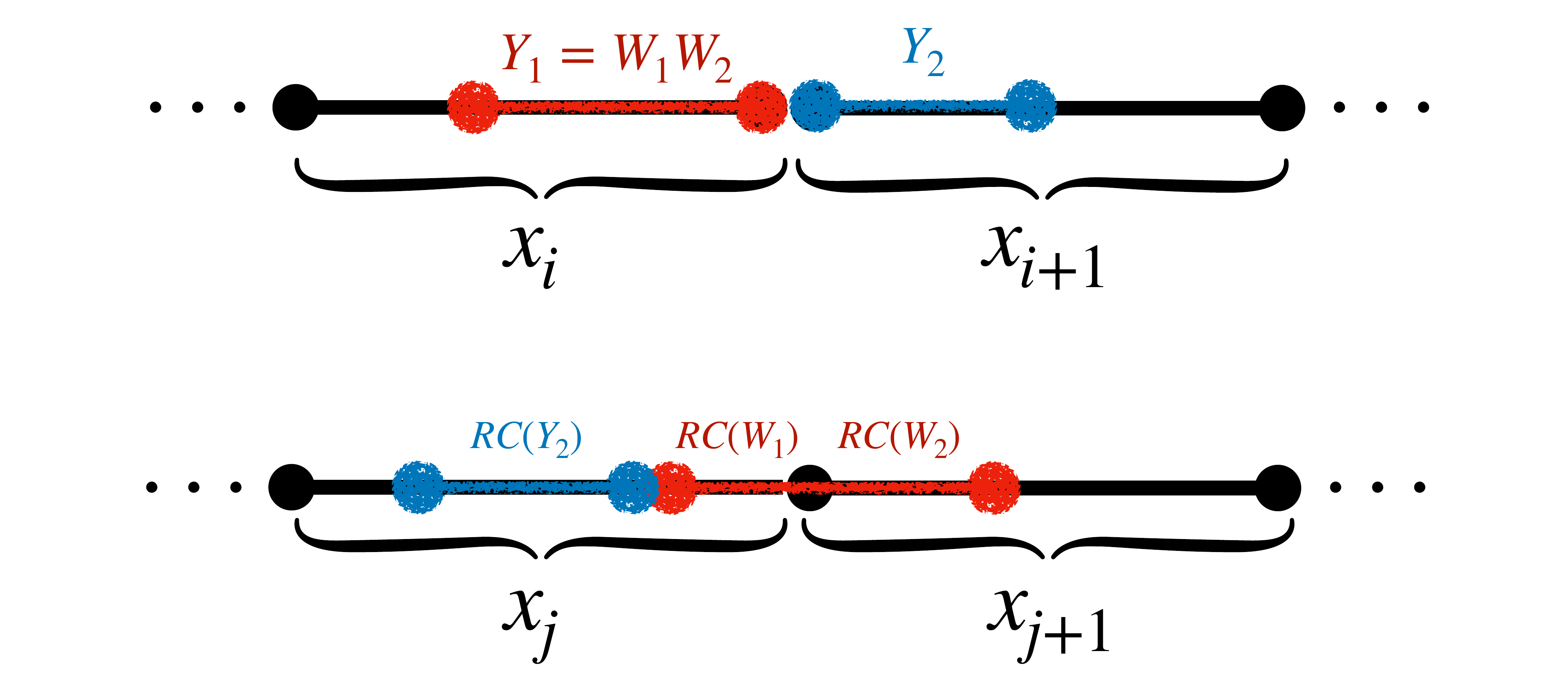}
\end{center}
\caption*{Consequently, $|Y_2| \ge t$, and we have $(\bx_{i+1},\bx_j)$ form a pair in $S_m^*$ that violate the condition.}
\caption{Proof Sketch of Theorem~\ref{con1}.}
\label{fig2}
\end{figure}
\begin{remark} Observe that, the set $S_m^*$ can be constructed via exhaustive search with complexity $O(2^m)$. In Section IV, we show that when $m$ is sufficiently large, $m\ge 3\log n+4=\Theta(\log n)$, there exists an efficient encoding/decoding algorithm for $(n,\D;m)$ SSA codes with at most one redundant symbol. Hence, for the case $m=o(\log n)$, we can use Construction~\ref{fixed-concatenate} to construct $(n,\D;m)$ SSA codes with complexity $2^m=\Theta(n)$. 
\end{remark} 

\subsection{Constructions via Symbol-Composition Constrained Codes}
In this subsection, we present an efficient construction for $(n,\D;m)$ SSA codes 
by simply restricting the symbol-composition for every subsequence of length $m$. Particularly, when $m = 3$, our method yields a family of DNA codes of rate $1.3031$ bits/nt, which is higher than the code rate in \cite{K:2021}. 
\vspace{0.05in}

\noindent{\em High Level Description}.
We select a nucleotide $x \in \D=\{{\bf A}, {\bf T}, {\bf C}, {\bf G}\}$, and let $y=\overline{x}\in \D$. For some $0<k\le m$, we present an efficient method to construct an $(n,\D;m)$ SSA code $\C$ as follows. For every codeword $\bc\in\C$, every subsequence $\bz$ of length $m$ of $\bc$ contains at least $k$ symbols $x$ while $\bz$ contains at most $(k-1)$ symbols $y$. We refer such a constraint to as the {\em symbol-composition constraint}. 
It is easy to verify that such a constructed code $\C$ is an $(n,\D;m)$ SSA code. Clearly, suppose on the other hand, there exists a pair of subsequences $\bz_1, \bz_2$ of length $\ell\ge m$ in $\bc\in \C$, such that $\bz_2={\rm RC}(\bz_1)$. It implies that there exists two subsequences of length $m$, which are $\bz_1'$ of $\bz_1$ and $\bz_2'$ of $\bz_2$, and $\bz_2'={\rm RC}(\bz_1')$. Since $\bz_1'$ contains at least $k$ symbols $x$, we have $\bz_2'={\rm RC}(\bz_1')$ must contain at least $k$ symbols $y=\overline{x}$. We then have a contradiction.
\vspace{0.05in}

The following construction is for $m=3$ and $k=1$.

\begin{construction}[Symbol-Composition Constrained Codes for $m=3$, $k=1$]\label{con2}
Given $n>0$, we select $x={\bf A}$ and $y=\overline{x}={\bf T}$. Set $\D^*=\{{\bf A}, {\bf C}, {\bf G}\}$. Let $\C_n$ be the set of all DNA sequences of length $n$ from alphabet $\D^*$ such that for any $\bc\in\C_n$, every subsequence of length three of $\bc$ must contain an ${\bf A}$. 
\end{construction}

\begin{theorem} We have $|\C_1|=3, |\C_2|=9,|\C_3|=19$, and
\begin{equation*}
|\C_n|= |\C_{n-1}| + 2| \C_{n-2}| + 4 |\C_{n-3}|.
\end{equation*}  
In addition, $\C_n$ is an $(n,\D;3)$ SSA code for all $n>0$. The asymptotic rate of this code family is given by $\log (\lambda) \approx 1.3031$, where $\lambda\approx 2.4675$ is the largest real root of $x^3-x^2-2x-4=0$.
\end{theorem}

\begin{proof}
Consider the code $\C_n$. For a codeword $\bc \in \C_n$, for any subsequence $\bx$ of length $\ell \ge 3$ of $\bc$, we have $\bx$ includes ${\bf A}$. On the other hand, since $\overline{\bf A}={\bf T}$ is not used in $\bc$, there is no reverse-complement of $\bx$ in $\bc$. In conclusion, $\bc$ is 3-SSA, or $\C_n$ is an $(n,\D;3)$ SSA code. 

We now prove the cardinality of  $\C_n$. it is easy to verify that $|\C_1|=3, |\C_2|=9,|\C_3|=19.$ For $n\ge 4$, we construct $\C_n$ recursively as
follows: 
\begin{small}
\begin{align*}
S^1_n=& \{ \bx {\bf A}: \text{ for } \bx \in \C_{n-1}\} \\
S^2_n=& \{ \bx {\bf A} {\bf C},  \bx {\bf A} {\bf G}: \text{ for } \bx \in \C_{n-2}\} \\
S^3_n=& \{ \bx {\bf A} {\bf C} {\bf C}, \bx {\bf A} {\bf C} {\bf G}, \bx {\bf A} {\bf G} {\bf C}, \bx {\bf A} {\bf G} {\bf G}: \text{ for } \bx \in \C_{n-3}\}, \text{and} \\
\C_n =& S^1_n \cup S^2_n \cup S^3_n.
\end{align*}
\end{small}
In other words, $S^1_n$ is the set formed by concatenating all sequences in $\C_{n-1}$ with ${\bf A}$, $S^2_n$ is the set formed by concatenating all sequences in $\C_{n-2}$ with ${\bf A}{\bf C}$ or ${\bf A}{\bf G}$, and lastly, $S^2_n$ is the set formed by concatenating all sequences in $\C_{n-3}$ with ${\bf A} {\bf C} {\bf C}, {\bf A} {\bf C} {\bf G}, {\bf A} {\bf G} {\bf C},$ or ${\bf A} {\bf G} {\bf G}$. It is easy to verify that $S^i_n \cap S^j_n\equiv \emptyset$, and the union $S^1_n \cup S^2_n \cup S^3_n$ includes all possible sequences in $\C_n$. Therefore, we have $|\C_n|= |\C_{n-1}| + 2| \C_{n-2}| + 4 |\C_{n-3}|.$
\end{proof}

Construction~\ref{con2} can be generalized to construct $(n,\D;m)$ SSA codes with $k=1$ as follows. 

\begin{theorem}[Symbol-Composition Constrained Codes for General $m$, $k=1$]\label{con3}
Given $n,m>0$. Set $\D^*=\{{\bf A}, {\bf C}, {\bf G}\}$, and $\C_n(m)$ to be the set of all sequences $\bx$ of length $n$ from alphabet $\D^*$ such that every subsequence of length $m$ of $\bx$ include an ${\bf A}$. We then have $|\C_i(m)|=3^i$ for $0\le i\le m-1$, and 
\begin{equation*}
|\C_n(m)|=\sum_{j=0}^{m-1} 2^j |\C_{n-j-1}(m)| \text{ for } n\ge m.
\end{equation*}
We then have $\C_n(m)$ is an $(n,\D;m)$ SSA code for all $n>0$. The asymptotic rate of this code family is given by $\log (\lambda)$, where $\lambda$ is the largest real root of $x^m-\sum_{j=0}^{m-1} 2^j x^{m-j}=0$.  
\end{theorem} 

\begin{remark} In general, given $m>k>0$, set $x={\bf A}$ and $y=\overline{x}={\bf T}$. we use $\C_n(m,k)$ to denote the set of all sequences $\bc\in\D^n$ such that every subsequence $\bz$ of length $m$ of $\bc$ contains at least $k$ symbols ${\bf A}$ while $\bz$ contains at most $(k-1)$ symbols ${\bf T}$. As shown earlier, $\C_n(m,k)$ is an $(n,\D;m)$ SSA code for all $m,k$. A natural question is, for a given number $m>0$, what is the value of $k$, where $1\le k\le m$, such that the code $\C_n(m,k)$ has the largest cardinality? We defer the study of $\C_n(m,k)$, including the code's cardinality and the design of efficient encoding algorithms to map arbitrary DNA sequences into such a code, to future research work. 
\end{remark}

\section{Constructions of $(n,\D;m)$ SSA Codes for $m\ge 3\log n+4$ with One Redundant Symbol} 

In this section, we show that when the stem length is sufficiently large, $m \ge 3\log n+4=\Theta(\log n)$, there exists an efficient encoding/decoding algorithm for $(n,\D;m)$ SSA codes with at most one redundant symbol. For simplicity, we assume that $\log_4 n$ is an integer, and define the {\em DNA-representation} of an integer as follows.  


\begin{definition}
For a positive integer $N$, the {\em DNA-representation} of $N$ is the replacement of symbols in the quaternary representation of $N$ over $\Sigma_4 = \{0,1,2,3\}$ by the following rule:  $0  \leftrightarrow {\bf A}, 1  \leftrightarrow {\bf T}, 2  \leftrightarrow {\bf C}, \text{ and } 3  \leftrightarrow {\bf G}.$
\end{definition} 

\begin{example}
If $N=100$, the quaternary representation of length 4 of $N$ is $1210$, hence, the DNA-representation of $N$ is ${\bf T}{\bf C}{\bf T}{\bf A}$. Similarly, when $N=55$, the quaternary representation of length 4 of $N$ is $0313$, thus the DNA-representation of $N$ is ${\bf A}{\bf G}{\bf T}{\bf G}$.
\end{example}

We now present explicit construction of the encoder $\enc_{{\rm SSA}}$ and the corresponding decoder $\dec_{\rm SSA}$. Our method is based on the sequence replacement technique. This method has been widely used in the literature \cite{TT:2020, TT:2022, srt:2019}. In addition, we also restrict the length of the repeated patterns of size 2 (also known as {\em pattern length limited (PLL) constraint}, as introduced in \cite{{wang:2021}}).
\vspace{0.05in}

\noindent{\bf Construction of $\enc_{{\rm SSA}}$}. Given $n>m>0$, $n> 16$, and $m\ge 3\log n+4$. Set $m'=1.5\log n+2$. The source DNA sequence $\bx \in \D^{n-1}$. The encoding algorithm includes three phases: {\em prepending phase}, {\em scanning and replacing phase}, and {\em extending phase}.
\vspace{0.05in}

\noindent{\em Prepending phase}. The source sequence $\bx\in\D^{n-1}$ is prepended with ${\bf A}$, to obtain $\bc={\bf A}\bx$ of length $n$. If $\bc$ is an $m$-SSA sequence, then the encoder outputs $\bc$. Otherwise, it proceeds to the next phase. 
\vspace{0.05in} 

\noindent{\em Scanning and replacing phase}. The encoder searches for the first pair of non-overlapping subsequences $\by, \bz$ of length $\ell_1$ of $\bc$, where $\ell_1\ge m'$, such that $\by={\rm RC}(\bz)$, or the first subsequence $\bu$ of $\bc$ of the form $\bu=(x_1x_2)^t$ whose length is $\ell_2=2t \ge m'=1.5\log n+2$, where $x_1,x_2\in  \D=\{{\bf A},{\bf T},{\bf C},{\bf G}\}$.

\begin{itemize}
\item If it finds a pair of non-overlapping subsequences $\by, \bz$, suppose that $\bc=\bX_1 \by \bX_2 \bz \bX_3$, where $\bX_1, \bX_2, \bX_3$ are subsequences of $\bc$, and $\by$ starts at index $i$, ends at index $j$ in $\bc$, where $j=i+\ell_1-1$, and $\bz$ starts at index $k$ in $\bc$. We have $i,j,k \le n-1$. 
\vspace{0.05in}

{\em Type-I Replacement}. The encoder sets a pointer ${\rm P_I}$, starting with symbol ${\bf T}$, and ${\rm P_I}={\bf T} \bp_1 \bp_2 \bp_3$, where $\bp_1, \bp_2, \bp_3$ are the DNA-representation of $i, j,$ and $k$, respectively. Since $\bp_1, \bp_2, \bp_3$ are of length $\log_4 n$, the pointer sequence ${\rm P}_{\rm I}$ is of length $1+3\log_4 n=1+1.5 \log n$. It then removes $\bz$ from $\bc$ and prepends ${\rm P_I}$ to $\bc$. The replacing step can be illustrated as follows. 
\begin{equation*}
 \bX_1 {\color{blue}{\by}} \bX_2 {\color{blue}{\bz}} \bX_3 \to \bX_1 {\color{blue}{\by}} \bX_2 \bX_3 \to {\color{red}{{\bf T} \bp_1 \bp_2 \bp_3}} \bX_1 {\color{blue}{\by}} \bX_2 \bX_3
\end{equation*}
Noted that the removed sequence $\bz$ is of length $\ell_1 \ge m' = 1.5 \log n+2$, while the insertion pointer ${\rm P}_I$ is of length $1.5 \log n+1$. Consequently, such a replacement reduces the length of the current sequence by at least one symbol. 

\item On the other hand, suppose that it finds a subsequence $\bu$ of $\bc$ of the form $\bu=(x_1x_2)^t$ whose length is $\ell_2=2t \ge m'$, where $x_1,x_2\in  \D=\{{\bf A},{\bf T},{\bf C},{\bf G}\}$. We further suppose that $\bc=\bU_1 (x_1x_2)^{t} \bU_2$, where $\bU_1, \bU_2$ are subsequences of $\bc$, and $\bu$ starts at index $i$, and ends at index $j$ in $\bc$, where $j=i+\ell_2-1$. We have $i,j \le n-1$.
\vspace{0.05in}

{\em Type-II Replacement}. Similarly, the encoder sets a pointer ${\rm P_{II}}$, starting with symbol ${\bf C}$, and ${\rm P_{II}}={\bf C} x_1 x_2 \bq_1 \bq_2$, where $\bq_1, \bq_2$ are the DNA-representation of $i$ and $j$, respectively. Since $\bq_1, \bq_2$ are of length $\log_4 n$, the pointer sequence ${\rm P_{II}}$ is of length $1+2+2\log_4 n=3+\log n$. It then removes $(x_1x_2)^{\ell_2}$ from $\bc$ and prepends ${\rm P_{II}}$ to $\bc$. The replacing step can be illustrated as follows. 
\begin{equation*}
 \bU_1 {\color{blue}{(x_1x_2)^{t}}} \bU_2 \to \bU_1 \bU_2  \to {\color{red}{{\bf C} x_1 x_2 \bq_1 \bq_2}} \bU_1 \bU_2.
\end{equation*}
Noted that the removed sequence is of length $\ell_2 \ge m' = 1.5 \log n+2$, while the insertion pointer ${\rm P_{II}}$ is of length $\log n+3$. Hence, such a replacement reduces the length of the current sequence by at least $(0.5\log n-1)$ symbols. Observe that $0.5\log n-1>1$ for $n> 16$.
\end{itemize}

\noindent The encoder repeats the scanning and replacing steps until the current sequence $\bc$ contains no pair of non-overlapping subsequences of length more than or equal to $m'$ such that one is the reverse-complement of the other, no subsequence $\bu$ of the form $\bu=(x_1x_2)^t$ whose length is $\ell_2=2t \ge m'$, or the current sequence is of length $m'-1$. Note that each replacement (either Type-I or Type-II) reduces the length of the current sequence by at least one symbol, and hence, this procedure is guaranteed to terminate. Here, we also note that the order of the scanning step is defined according to the starting index of the corresponding subsequences. In case the first subsequence $\by$ forming a secondary structure, is also the starting of such a subsequence $\bu$, the encoder proceeds to type-I replacement.  
\vspace{0.05in} 

\noindent{\bf Extending phase}. If the length of the current sequence $\bc$ is $N_0$ where $N_0<n$, the encoder appends a suffix of length $N_1=n-N_0$ to obtain a sequence of length $n$. Surprisingly, regardless the choice of the appending suffix, there is an efficient algorithm to decode the source DNA sequence uniquely (refer to the construction of $\dec_{{\rm SSA}}$). Here, we present one efficient method to generate a suitable suffix so that the output codeword remains $m$-SSA.  
\begin{itemize}
\item If $N_1$ is even, we append $\bs=({\bf A}{\bf C})^{N_1/2}$ to the end of $\bc$. 
\item If $N_1$ is odd, we append $\bs=({\bf A}{\bf C})^{(N_1-1)/2}{\bf A}$ to the end of $\bc$.
\end{itemize}

\begin{theorem}
The encoder $\enc_{{\rm SSA}}$ is correct. In other words, $\enc_{{\rm SSA}}(\bx)$ is an $m$-SSA sequence of length $n$ for all $\bx\in\D^{n-1}$. The redundancy of $\enc_{{\rm SSA}}$ is one redundant symbol. 
\end{theorem}
\begin{proof}
Suppose that $\bc=\enc_{{\rm SSA}}(\bx)\in \D^n$, and $\bc=\bc_1 \bs$, where $\bc_1$ is $m'$-SSA and the length of the repeated patterns of size 2 in $\bc_1$ is of length at most $m'=1.5\log n+2$, and $\bs$ is the generated suffix of $\bc_1$ at the extending phase. Consider an arbitrary sequence $\by$ of length $\ell \ge 3\log n+4$. Suppose that $\by=\by_1\by_2$, where $\by_1$ is a subsequence of $\bc_1$ while $\by_2$ is a subsequence of $\bs$. We have the following cases. 
\begin{itemize}
\item If $\by_1$ is of length less than $m'$ (particularly including the case $\by_1\equiv \varnothing$), hence the length of $\by_2$ is more than $m'$. Clearly, there is no subsequence $\bz$ in $\bc_1 \bs$ that $\by={\rm RC}(\bz)$, as the length of the repeated patterns of size 2 in $\bc_1$ is of length at most $m'$. 
\item If $\by_1$ is of length more than or equal to $m'$, we also conclude that there is no subsequence $\bz$ in $\bc=\bc_1 \bz$ that $\by={\rm RC}(\bz)$ since $\bc_1$ is $m'$-SSA. \qedhere
\end{itemize}
\end{proof}
We now present the corresponding decoding algorithm. 
\vspace{0.05in} 

\noindent{\bf Construction of $\dec_{{\rm SSA}}$}. From a DNA sequence $\bc$ of length $n$, the decoder scans from left to right. If the first symbol is ${\bf A}$, the decoder simply removes ${\bf A}$ and identifies the last $(n-1)$ symbols as the source sequence. On the other hand, 
\begin{itemize} 
\item if it starts with ${\bf T}$, the decoder takes the prefix of length $(1+1.5\log n)$ and concludes that this prefix is a pointer prepended after a type-I replacement. In other words, the pointer is of the form ${\bf T}\bp_1\bp_2\bp_3$, where $\bp_1,\bp_2,\bp_3$, each is of length $\log_4 n=0.5\log n$. The decoder sets $i, j, k$ to be the positive integers whose DNA-representations are $\bp_1,\bp_2,\bp_3$, respectively and sets $\by$ to be the subsequence containing the symbols from index $i$ to index $j$. It removes the pointer, adds $\bz \equiv {\rm RC}(\by)$ to $\bc$ at index $k$. 
\item if it starts with ${\bf C}$, the decoder takes the prefix of length $(3+\log n)$ and concludes that this prefix is a pointer prepended after a type-II replacement. In other words, the pointer is of the form ${\bf C}x_1x_2\bq_1\bq_2$, where $\bq_1,\bq_2$, each is of length $\log_4 n=0.5\log n$. The decoder sets $i, j$ to be the positive integers whose DNA-representations are $\bq_1,\bq_2$, respectively. It then removes the pointer, adds $\bz \equiv (x_1x_2)^{(j-i+1)/2}$ to $\bc$ at index $i$. 
\end{itemize}
The decoding procedure terminates when the first symbol is {\bf A}, and takes the following $(n-1)$ symbols as the user data.
\vspace{0.05in} 

\noindent{\bf Complexity analysis}. For codeword of length $n$, the time complexity of the encoder (and the corresponding decoder) is linear in $n$, which follows from: the number of replacing operations is at most $n-m$, which is $\Theta(n)$, and the complexity of the each replacing operation (including the prepending prefix step or converting quaternary representation to DNA-representation of an integer) is constant time $\Theta(1)$. 

\section{Conclusion}

We have presented efficient algorithms to construct DNA codes that avoid secondary structure of arbitrary stem length. Particularly, when $m\ge 3 \log n + 4$, we have provided an efficient encoder that incurs only one redundant symbol, and when $m = 3$, our constructions yield a family of DNA codes of rate $1.3031$ bits/nt, that improve the previous highest code rate in the literature. 

\label{sec:biblio}


\begin{thebibliography}{99}

\bibitem{adle:1994} L. M. Adleman, ``Molecular computation of solutions to combinatorial problems," {\em Science}, vol. 266, pp. 1021-1024, Nov. 1994.

\bibitem{church2012} G. M. Church, Y. Gao, and S. Kosuri, ``Next-generation digital information storage in DNA," {\em Science}, vol. 337, no. 6102, pp. 1628-1628, 2012.

\bibitem{fountain2017} Y. Erlich and D. Zielinski, ``DNA fountain enables a robust and efficient storage architecture," {\em Science}, vol. 355, no. 6328, pp. 950-954, 2017.


\bibitem{Organick:2018}
L.~Organick, S.~Ang, Y.~J.~Chen, R.~Lopez, S.~Yekhanin, K.~Makarychev, M.~Racz, G.~Kamath, P.~Gopalan, B.~Nguyen, C.~Takahashi, S.~Newman, H.~Y.~Parker, C.~Rashtchian, K.~Stewart, G.~Gupta, R.~Carlson, J.~ Mulligan, D.~Carmean, G.~Seelig, L.~Ceze, and K.~Strauss,
``Random access in large-scale {DNA} data storage",
\emph{Nature Biotechnology},
vol.~36, 242--248, 2018.

\bibitem{goldman2013} N. Goldman, P. Bertone, S. Chen, C. Dessimoz, E. M. LeProust, B. Sipos, and E. Birney, ``Towards practical, high-capacity, low-maintenance information storage in synthesized DNA," {\em Nature}, vol. 494, 77-80, 2013.



\bibitem{son:2004} Y. Benenson, B. Gil, U. Ben-Dor, R. Adar and E. Shapiro, ``An autonomous molecular computer for logical control of gene expression," {\em Nature}, vol. 429, pp. 423-429, May 2004.

\bibitem{yazdi:2017}S.~M.~H.~T.~Yazdi, S. M., R.~Gabrys and O. Milenkovic, ``Portable and error-free DNA-based data storage," 
{\em Scientific reports}, 7(1), 1-6, 2017.


\bibitem{O2:2005} O. Milenkovic and N. Kashyap, ``On the design of codes for DNA computing," in {\em Coding Cryptogr.}, Germany: Springer, Mar. 2006, pp. 100-119.

\bibitem{yazdi:2018}
S.~M.~H.~T.~Yazdi, H.~M.~Kiah, H. M., R.~Gabrys, and O. Milenkovic, 
``Mutually uncorrelated primers for DNA-based data storage,'' 
\emph{IEEE Transactions on Information Theory}, 64(9), 6283-6296, 2018.



\bibitem{chee:2020}
Y.~M.~Chee, H.~M.~Kiah and H.~Wei,
``Efficient and explicit balanced primer codes,'' 
\emph{IEEE Transactions on Information Theory}, 66(9), 5344--5357, 2020.


\bibitem{K:2021} K. G. Benerjee and A. Banerjee, ``On DNA Codes With Multiple Constraints," in {\em IEEE Communications Letters}, vol. 25, no. 2, pp. 365-368, Feb. 2021, doi: 10.1109/LCOMM.2020.3029071.

\bibitem{KDG:2021} K. G. Benerjee, S. Deb, and M. K. Gupta, ``On conflict free DNA codes", {\em Cryptogr. Commun.} {\bf 13}, 143-171, 2021. https://doi.org/10.1007/s12095-020-00459-7.

\bibitem{O:2005} O. Milenkovic and N. Kashyap, ``DNA codes that avoid secondary structures," in {\em Proceedings. International Symposium on Information Theory}, Sep. 2005, pp. 288-292.

\bibitem{TT:2021} T. T. Nguyen, K. Cai, K. A. Schouhamer Immink and H. M. Kiah, ``Capacity-Approaching Constrained Codes With Error Correction for DNA-Based Data Storage," in {\em IEEE Transactions on Information Theory}, vol. 67, no. 8, pp. 5602-5613, Aug. 2021, doi: 10.1109/TIT.2021.3066430.

\bibitem{KING:2003} O.D. King, ``Bounds for DNA codes with constant GC-content," {\em The Electronic Journal of Combinatorics}, vol. 10, no. 1, R33, 2003.

\bibitem{KING:2005} P. Gaborit and O.D. King, ``Linear constructions for DNA codes," {\em Theoretical Computer Science}, vol. 334, no. 1-3, pp. 99-113, April 2005.

\bibitem{segmented} K. Cai, H. M. Kiah, M. Motani and T. T. Nguyen, ``Coding for Segmented Edits with Local Weight Constraints," {\em 2021 IEEE International Symposium on Information Theory (ISIT)}, 2021, pp. 1694-1699, doi: 10.1109/ISIT45174.2021.9517851.


\bibitem{Bres:1986} K. Breslauer, R. Frank, H. Blocker, and L. Marky, ``Predicting DNA duplex stability from the base sequence," {\em Proc. Natl. Acad. Sci.} USA, vol. 83, pp. 3746-3750, 1986.

\bibitem{RN:1980}  R. Nussinov and A.B. Jacobson, ``Fast algorithms for predicting the secondary structure of single stranded RNA ," {\em Proc. Natl. Acad. Sci.}, USA, vol. 77, no. 11, pp. 6309-6313, 1980.

\bibitem{heck:2019}  R. Heckel, G. Mikutis, and R. N. Grass, ``A characterization of the DNA data storage channel," {\em Sci. Rep.}, vol. 9, no. 1, pp. 1-12, Jul. 2019.

\bibitem{TT:2020} T. Thanh Nguyen, K. Cai and K. A. Schouhamer Immink, ``Binary Subblock Energy-Constrained Codes: Knuth's Balancing and Sequence Replacement Techniques," {\em 2020 IEEE International Symposium on Information Theory (ISIT)}, Los Angeles, CA, USA, 2020, pp. 37-41, doi: 10.1109/ISIT44484.2020.9174430. 


\bibitem{srt:2019} O. Elishco, R. Gabrys, M. Medard, and E. Yaakobi, ``Repeated-Free Codes", {\em Proc. IEEE Int. Symp. Inf. Theory (ISIT)}, Paris, France, 2019.

\bibitem{TT:2022} T. T. Nguyen, K. Cai, H. M. Kiah, K. A. Schouhamer Immink and Y. M. Chee, ``Using One Redundant Bit to Construct Two-Dimensional Almost-Balanced Codes," {\em 2022 IEEE International Symposium on Information Theory (ISIT)}, 2022, pp. 3091-3096, doi: 10.1109/ISIT50566.2022.9834724.


\bibitem{wang:2021} S. Wang, J. Sima and F. Farnoud, ``Non-binary Codes for Correcting a Burst of at Most 2 Deletions," {\em 2021 IEEE International Symposium on Information Theory (ISIT)}, Melbourne, Australia, 2021, pp. 2804-2809, doi: 10.1109/ISIT45174.2021.9517917.

\end{thebibliography}
\end{document}